\documentclass[11pt]{article}
\usepackage{amsmath}
\usepackage{amsmath}
\usepackage{times}
\usepackage{subfigure}
\usepackage{fullpage}
\usepackage{color}
\usepackage{graphicx,amssymb,amsmath}
\usepackage{multirow}
\usepackage{xspace}
\usepackage[linesnumbered,boxed,ruled, vlined]{algorithm2e} 
\usepackage[top=1in, bottom=1in, left=1in, right=1in]{geometry}
\usepackage{cases} 

\usepackage{url}
\urldef{\mailsa}\path|{huanglx12@mails, lijian83@mail, sqc12@mails}.tsinghua.edu.cn|

\newcommand{\commentout}[1]{}
\newcommand{\eat}[1]{}
\newcommand{\topic}[1]{\vspace{0.2cm}\noindent{\bf #1 :}}

\newcommand{\calH}{{\mathcal H}}
\newcommand{\calJ}{{\mathcal J}}
\newcommand{\calC}{{\mathcal C}}
\newcommand{\calF}{{\mathcal F}}
\newcommand{\calT}{{\mathcal T}}
\newcommand{\calP}{{\mathcal P}}
\newcommand{\calG}{{\mathcal G}}
\newcommand{\calU}{{\mathcal U}}
\newcommand{\calV}{{\mathcal V}}

\newcommand{\calS}{{\mathcal S}}

\newcommand{\opt}{\mathsf{Opt}}
\newcommand{\OPT}{\mathsf{OPT}}

\newcommand{\sol}{\mathsf{Sol}}

\newcommand{\dist}{\mathsf{dist}}

\newcommand{\Z}{\mathbb{Z}}
\newcommand{\R}{\mathbb{R}}

\newcommand{\hx}{\widehat{x}}
\newcommand{\hy}{\widehat{y}}
\newcommand{\tx}{\widetilde{x}}
\newcommand{\ty}{\widetilde{y}}
\newcommand{\bx}{\overline{x}}

\newcommand{\hp}{\widehat{p}}

\newcommand{\rednote}[1]{#1}

\newtheorem{assumption}{Assumption}

\newcommand{\mincsc}{$\mathsf{MIN}$-$\mathsf{CSC}$}
\newcommand{\minwcsc}{$\mathsf{MIN}$-$\mathsf{WCSC}$}

\newcommand{\mincds}{$\mathsf{MIN}$-$\mathsf{CDS}$}
\newcommand{\minds}{$\mathsf{MIN}$-$\mathsf{DS}$}
\newcommand{\minwcds}{$\mathsf{MIN}$-$\mathsf{WCDS}$}
\newcommand{\minwds}{$\mathsf{MIN}$-$\mathsf{WDS}$}

\newcommand{\bcsc}{$\mathsf{Budgeted}$-$\mathsf{CSC}$}

\newcommand{\bsc}{$\mathsf{MaxCov}$}
\newcommand{\qst}{$\mathsf{QST}$}
\newcommand{\gst}{$\mathsf{GST}$}
\newcommand{\hittingset}{$\mathsf{HitSet}$}
\newcommand{\lpr}{\mathsf{Lp}\text{-}\mathsf{GST}}
\newcommand{\lpflow}{\mathsf{Lp}\text{-}\mathsf{flow}}
\newcommand{\lphittingset}{\mathsf{Lp}\text{-}\mathsf{HS}}
\newcommand{\lpsteiner}{\mathsf{Lp}\text{-}\mathsf{ST}}
\newcommand{\group}{\mathsf{gp}}
\newcommand{\cell}{\mathsf{cl}}
\newcommand{\setcell}{\mathsf{\Delta}}

\newcommand{\cells}{\mathsf{CG}}
\newcommand{\terminal}{\mathsf{Ter}}
\renewcommand{\mod}{\operatorname{mod}}

\newcommand{\Rc}{R_\mathsf{c}}
\newcommand{\Rs}{R_\mathsf{s}}
\newcommand{\Dc}{D_\mathsf{c}}
\newcommand{\Ds}{D_\mathsf{s}}
\newcommand{\Gc}{\mathcal{G}_\mathsf{c}}

\newenvironment{proof}{\noindent {\em Proof: }\ignorespaces}{}

\newcommand{\qed}{\hspace*{\fill}$\Box$\medskip}

\newtheorem{theorem}{Theorem}

\newtheorem{lemma}{Lemma}

\newtheorem{definition}{Definition}

\newtheorem{observation}{Observation}

\newtheorem{remark}{Remark}

\title{Approximation Algorithms for the Connected Sensor Cover Problem}

\author{Lingxiao Huang \quad\quad\quad\quad Jian Li \thanks{Corresponding author.}
	\quad \quad \quad \quad Qicai Shi  \\
Institute for Interdisciplinary Information Sciences\\
Tsinghua University, China \\
\mailsa
}

\begin{document}

\maketitle

\begin{abstract}
	We study the minimum connected sensor cover problem  (\mincsc) and
	the budgeted connected sensor cover  (\bcsc) problem,
both motivated by important applications (e.g., reduce the communication cost among sensors) in wireless sensor networks.
In both problems, we are given a set of sensors and a set of target points in the Euclidean plane.
In \mincsc, our goal is to find a set of sensors of minimum cardinality,
such that all target points are covered, and all sensors can communicate with each other (i.e., the communication graph is  connected). We obtain a constant factor approximation algorithm, assuming that the ratio between the sensor radius and communication radius is bounded.
In \bcsc\ problem, our goal is to choose a set of $B$ sensors, such that the number of targets covered by the chosen sensors is maximized and the communication graph is connected. We also obtain a constant approximation under the same assumption.
\end{abstract}

\vspace{-0.2cm}
\section{Introduction}
\vspace{-0.2cm}

In many applications, we would like to monitor a region or a collection of targets of interests
by deploying a set of wireless sensor nodes.
A key challenge in such applications is the limited energy supply for each sensor node.
Hence, designing efficient algorithms for minimizing energy consumption and maximizing the lifetime of the network
is an important problem in wireless sensor networks and many variations have been studied extensively. We refer interested readers to the book by Du and Wan~\cite{du2012connected}
for many algorithmic problems in this domain.

In this paper, we consider two important sensor coverage problems.
Now, we introduce some notations and formally define our problem.
We are given a set $\calS$ of $n$ sensors in $\R^d$.
All sensors in $\calS$ have the same communication range $\Rc$
and the same sensing range $\Rs$.
In other words, two sensors $s$ and $s'$ can communicate with each other
if $\dist(s,s')\leq \Rc$, and a target point $p$ can be covered by sensor $s$
if $\dist(p,s)\leq \Rs$.
We use $D(s,R)$ to denote the disk with radius $R$ centered at point $s$.
Let $\Dc(s)=D(s,\Rc)$ and $\Ds(s)=D(s,\Rs)$.

\begin{assumption}[Funke et al.~\cite{funke2007improved}]
	\label{ass:1}
	In this paper, we assume that $\Rs/\Rc$ can be upper bounded by a constant $C=O(1)$
	(i.e., $\Rs/\Rc\leq C$). 
	Without loss of generality, we can assume that $\Rc=1$.
	Hence, $\Rs= O(1)$.
\end{assumption}

Note that this assumption holds for most practical
applications, e.g., it generalizes Funke et al.~\cite{funke2007improved} which assumes that $\Rs/\Rc\leq 1/2$.

The first problem we study is the {\em minimum Connected sensor covering} (\mincsc) problem.
This problem considers the problem of selecting the minimum number of sensors
that form a connected network and detect all the targets.
It is somewhat similar, but different from, the connected dominating set problem.
We will discuss the difference shortly.
The formal problem definition is as follows:

\begin{definition}
	\mincsc: Given a set $\calS$ of sensors
	and a set $\calP$ of target points,
	find a subset $\calS'\subseteq \calS$ of minimum cardinality
	such that all points in $\calP$ are covered by the union of sensor areas in $\calS'$
	and the communication links between sensors in $\calS'$ form a connected graph.
\end{definition}
\vspace{-0.1cm}

In some applications, instead of monitoring a set of discrete target points,
we would like to monitor a continuous range $R$, such as a rectangular area.
Such problems can be easily converted into a \mincsc\ with discrete points,
by creating a target point (which we need to cover) in each cell of the arrangement
of the sensing disks $\{\Ds(s)\}_{s\in \calS}$ restricted in $R$.

The second problem
studied in this paper is the {\em Budgeted connected sensor cover}  (\bcsc) problem.
The problem setting is the same as \mincsc, except that
we have an upper bound on the number of sensors we can open, and the goal becomes
to maximize the number of covered targets.

\vspace{-0.1cm}
\begin{definition}
	\bcsc: Given a set $\calS$ of sensors , a set $\calP$ of target points and a positive integer $B$,
	find a subset $\calS'\subseteq \calS$
	such that $|\calS'| \leq B$ and the number of points in $\calP$ covered by the union of sensor areas in $\calS'$ is maximum and the communication links between sensors in $\calS'$ form a connected graph.
\end{definition}

\rednote{Note that in this paper we only consider the unweighted versions for both problems. 
	We leave the weighted versions as an interesting future direction.}

\vspace{-0.3cm}
\subsection{Previous Results and Our Contributions}

\subsubsection{\mincsc}
The \mincsc\ problem was first proposed by Gupta et al.~\cite{gupta2006connected}.
They gave an $O(r\ln n)$-approximation
($r$ is an upper bound of the hop-distance between any two sensors
having nonempty sensing intersections).
Wu et al.~\cite{wu2013approximations} give an $O(r)$-approximation algorithm.
\rednote{Then, Wu et al.~\cite{wu2016connected} improved the approximation factor to $3r+1+(3r-2)\varepsilon$, which is best approximation ratio known so far (in terms of $r$).}
If $\Rs\leq \Rc/2$, $r=1$ and the above result implies a constant approximation.
However, even
$\Rs$ is slightly larger than $\Rc/2$, $r$ may still be  arbitrarily large.
We also notice that if $r=O(1)$, we must have $\Rs/\Rc= O(1)$.
So Assumption~\ref{ass:1} is a weaker assumption than the assumption that $r=O(1)$.
\rednote{Funke et al.~\cite{funke2007improved} showed that the greedy algorithm that provides complete coverage has an approximation factor no better than $\Omega(\log n)$.}

\mincsc\ is in fact a special case of the {\em group Steiner tree} problem
(as also observed in Wu et al~\cite{wu2013approximations,wu2016minimum}).
In fact, this can be seen as follows:
consider the communication graph (the edges are the communication links).
For each target, we create a group which consists for all sensor nodes that can cover the
target. The goal is to find a minimum cost tree spanning all groups.\footnote{
	Notice that the group Steiner tree is edge-weighted but \mincsc\ is node-weighted.
	However, since all nodes have the same (unit) weight, the edge-weight and node-weight of a tree differ by at most 1.
}
Garg et al~\cite{garg1998polylogarithmic}, combined with the optimal
probabilistic tree embedding \cite{fakcharoenphol2003tight}, \rednote{obtained} an $O(\log^{3}n)$ factor approximation algorithm the group Steiner tree problem via LP rounding.  Chekuri et al.~\cite{chekuri2006greedy} \rednote{claimed} nearly the same approximation ratio using pure combinatorial method.

Our first main contribution is a constant factor approximation algorithm for
\mincsc\ under Assumption~\ref{ass:1}, improving on the aforementioned results.
Our improvement heavily rely on the geometry of the problem (which the group Steiner tree approach ignores).

\begin{theorem}
	There is a polynomial time approximation algorithm which can achieve
	an approximation factor $O(C^2)$ for \mincsc.
	Under Assumption~\ref{ass:1}, the approximation factor is a constant.
\end{theorem}

\begin{remark}
	The weighted version of the connected sensor covering problem (\minwcsc) has also been studied, in which each sensor has a nonnegative weight and the goal is to find a set of minimum weight. 
	Elbassioni et al.~\cite{elbassioni2012relation} showed that the problem is also a special case of the group Steiner tree problem and claimed an $O(\sqrt{n} \log n)$ factor approximation algorithm.
\end{remark}

\vspace{-0.3cm}
\subsubsection{\bcsc}
Recall in \bcsc, we have a budget $B$, which is the upper bound of the number of sensors we can use and our goal is to maximize the number of covered target points.
Kuo et al.\cite{kuo2013maximizing} study this problem
under the assumption that the communication and the sensing radius of sensors are the same
(i.e., $\Rs=\Rc$). They obtained an $O(\sqrt{B})$-approximation by transforming the problem to a more general connected submodular function maximization problem.

Recently, Khuller et al.~\cite{khuller2014analyzing} obtained
a constant approximation for the {\em budgeted generalized connected dominating set problem},
defined as follows: Given an undirected graph $G(V,E)$ and budget $B$, and a  monotone {\em special submodular function}
\footnote{
	$f$ is a special submodular function if (1) $f$ is submodular:
	$f(A\cup\{v\})-f(A)\geq f(B\cup \{v\})-f(B)$ for any $A\subset B\subseteq V$;
	(2) $f(A\cup X)-f(A)=f(A\cup B\cup X)-f(A\cup B)$ if $N(X)\cap N(B)=\emptyset$
	for any $X,A,B\subseteq V$. Here, $N(X)$ denotes the neighborhood of $X$ (including $X$).
} $f: 2^V\rightarrow \Z^+$, find a subset $S\subseteq V$ such that $|S|\leq B$,
$S$ induces a connected subgraph
and $f(S)$ is maximized.
If $\Rs\leq \Rc/2$ in \bcsc, the coverage function $f(S)$ (the number of targets covered by sensor set $S$) is a special submodular function.\footnote{\rednote{Consider $X,A,B\subseteq V$ satisfying that $N(X)\cap N(B)=\emptyset$. It implies that for any $x\in X$ and $y\in B$, $d(x,y)> R_c$. Since $R_S\leq R_c/2$, we have that $D_s(x)\cap D_s(y)=\emptyset$. Hence, $f(A\cup X)-f(A)=f(A\cup B\cup X)-f(A\cup B) = \left|y\in U: y\in (\cup_{x\in X} D_s(x))\setminus (\cup_{x\in A} D_s(x)) \right|$. It implies that $f(S)$ is a special submodular function.}} 
Hence, we have a constant approximation
for \bcsc\ when $\Rs\leq \Rc/2$.
When $\Rs> \Rc/2$, $f(S)$ may not be special submodular and the algorithm and analysis
in \cite{khuller2014analyzing} do not provide any approximation guarantee for \bcsc.

We note that it is also possible to adapt the greedy approach developed by group Steiner tree~\cite{chekuri2006greedy} and
polymatroid Steiner tree~\cite{calinescu2005polymatroid} to get polylogarithmic approximation
for \bcsc. However, it is unlikely that the approach can be made to achieve constant approximation
factors, and we omit the details.

In this paper, we improve the above results by presenting the first constant
factor approximation algorithm under the more general Assumption~\ref{ass:1}.

\begin{theorem}
	There is a polynomial time approximation algorithm which can achieve
	approximation factor of $\frac{1}{102C^2}$ for \bcsc.
	Under Assumption~\ref{ass:1}, the approximation factor is $O(1)$.
	\label{th_bcsc}
\end{theorem}

Our algorithm is inspired by \cite{khuller2014analyzing}.
In particular, we make crucial use of the geometry of the problem
to get around the issue required by \cite{khuller2014analyzing}
(i.e., the coverage function is required to be special submodular in their work).

\vspace{-0.2cm}
\subsection{Other Related Work}
\vspace{-0.1cm}

\mincsc\ is closely related to
the {\em minimum dominating set} (\minds) problem
and the {\em minimum connected dominating set} (\mincds) problem.
In fact, if the communication radius $\Rc$ is equal to the sensing radius $\Rs$ and the collection $\calS$ of sensors is equal to the collection $\calP$ of target points,
\mincsc\ is equivalent to \mincds.
%
In general graphs, \mincds\ inherits the inapproximability of set cover, so it is NP-hard to approximation \mincds\ within a factor of $\rho\ln n$ for any $\rho<1$ \cite{feige1998threshold,dinur2014analytical}.
Improving upon Klein and Ravi~\cite{klein1995nearly}, Guha and Khuller~\cite{guha1999improved} obtained a
$1.35\ln n$-approximation, which is the best result known for general graphs.

Lichtenstein~\cite{lichtenstein1982planar} proved that \mincds\ in unit disk graphs (UDG) is NP-hard (which also implies that \mincsc\ is NP-hard).
The first constant approximation algorithm for the unweighted \mincds\ problem in UDG
was \rednote{obtained} by Wan et al.\cite{wan2002distributed}.
This was later improved by Cheng et al.\cite{cheng2003polynomial}, who gave the first PTAS.
Many variants of \minds\ and \mincds, motivated by various applications in wireless sensor network, 
have been studied extensively. See \cite{du2012connected} for a comprehensive treatment.

For the weighted (connected) dominating set problem (\minwds\ and \minwcds), Amb{\"u}hl et al.~\cite{ambuhl2006constant} \rednote{provided} the first constant ratio approximation algorithms for both problems (the constants are 72 and 94 for \minwds\ and \minwcds\ respectively).
The constants were improved in a series of subsequent
papers~\cite{huang2009better,dai20095+,zou2011new,willsonbetter}. Recently, Li and Jin~\cite{li2015ptas} 
\rednote{obtained} the first PTAS
for \minwds\ and an improved constant approximation for \minwcds\ in UDG.

\bcsc\ is a special case of the submodular function maximization problem subject to a
cardinality constraint and a connectivity constraint.
Submodular maximization under cardinality constraint,
which generalizes the maximum coverage problem, is a classical combinatorial optimization problem and it is known the optimal approximation is $1-1/e$~\cite{nemhauser1978analysis,feige1998threshold}.
Submodular maximization under various more general combinatorial constraints (in particular, downward monotone set systems)
is a vibrant research area in theoretical computer science
and there have been a number of exciting new developments in the past few years
(see e.g., \cite{calinescu2011maximizing,vondrak2011submodular} and the references therein).
The connectivity constraint has also been considered in some previous work \cite{zhang2012complexity,kuo2013maximizing,khuller2014analyzing}, some of which we mentioned before.

\section{Preliminaries}
\label{sec:prel}

We need the following \emph{maximum coverage} (\bsc) in our algorithms.
\begin{definition}
	\bsc: Given a universe $U$ of elements and a family $\calS$ of subsets of $U$, and a positive integer $B$,
	find a subset $\calS'\subseteq \calS$
	such that $|\calS'| \leq B$ and the number of elements covered by
	$\cup_{S\in\calS'}S$ is maximized.
\end{definition}
We need to following well known result, by \cite{nemhauser1978analysis,hochbaum1998analysis}.
\begin{lemma}[Corollary 1.1 of  Hochbaum and Pathria \cite{hochbaum1998analysis}]
	The greedy algorithm is a $(1-\frac{1}{e})$-approximation for \bsc .
	\label{lm_bsc}
\end{lemma}

A closely related problem is the {\em hitting set} problem.
\begin{definition}
	\hittingset: Given a universe $U$ of weighted elements (with weight function $c:U\rightarrow\R^+$)
	and a family $\calS$ of subsets of $U$
	find a subset $H\subseteq U$
	such that $H\cap S \neq \emptyset$ for all $S\in \calS$ (i.e., $H$ hits every subset in $\calS$) and
	$\sum_{u\in H}c_u$ is minimized.
\end{definition}

The \hittingset\ problem is equivalent to the set cover problem (where the elements and subsets switch roles).
It is well known that a simple greedy algorithm
can achieve an approximation factor of $\ln n$ for \hittingset\
and the factor is essentially optimal \cite{feige1998threshold,dinur2014analytical}.
In this paper, we use a geometric version of \hittingset\ in which
the set of given elements are points in $\R^2$ and the subsets are
induced by given disks (i.e., each $S\in\calS$ is the subset of points that can be covered by
a given disk). Geometric hitting set admits constant factor approximation algorithms (even PTAS)
for many geometric objects (including disks) \cite{bronnimann1995almost,Clarkson,mustafa2009ptas,varadarajan2010weighted,Chan2012}.
As mentioned in the introduction, \mincsc\ is a special case of
the following {\em group Steiner tree} (\gst) problem.

\begin{definition}
	\gst: We are given an undirected graph $G=(V,E,c,\calF)$
	where $c: E\rightarrow \Z^+$ is the edge cost function,
	and $\calF$ is a collection of subsets of $V$.
	Each subset in $\calF$ is called a group.
	The goal is to find a subtree $T$, such that
	$T\cap S\ne \emptyset$ for all $S\in \calF$ (i.e., $T$ spans all groups)
	and the cost of the tree $\sum_{e\in T}c_e$ is minimized.
\end{definition}

Our algorithm for \bcsc\ also needs the following {\em quota Steiner tree} (\qst) problem.

\begin{definition}
	\qst: Given an undirected graph $G=(V,E,c,p)$
	($c: E\rightarrow \Z^+$ is the edge cost function,
	$p:V\rightarrow\Z^+$ is the vertex profit function) and an integer $q$,
	find a subtree $T=\mathop{\arg\min}_{T \subset E,\sum_{v\in T } p(v) \geq q} \sum_{e \in T} c(e)$ of the graph $G$
	($T$ tries to collect as much profit as possible subject to the quota constraint).
\end{definition}

Johnson et al.~\cite{johnson2000prize} proposed the \qst\ problem and proved that any $\alpha$-approximation for the $k$-$\mathsf{MST}$ problem yields an $\alpha$-approximation for the \qst\ problem.
Combining with the $2$-approximation for $k-\mathsf{MST}$ developed by Garg \cite{garg2005saving},
we can get a $2$-approximation for the \qst\ problem.
\begin{lemma}
	\label{lm:qst}
	These is an  approximation algorithm with approximation factor $2$
	for \qst.
\end{lemma}


\section{Minimum Connected Sensor Cover}
\label{sec:mscs}

We first construct an edge-weighted graph $\Gc$ as follows:
If $\dist(s,s')\leq \Rc$, we add an edge between $s$ and $s'$
(It is easy to see that $\Gc$ is in fact a unit disk graph).
$\Gc$ is called {\em  the communication graph}.
Recall that \mincsc\ requires us to find a set of vertices
that induces a connected subgraph in the communication graph $\Gc$.

First, we note that $\Gc$ may have several connected components.
We can see any feasible solution must be contained in a single connected component
(otherwise, the solution can not induce a connected graph).
Our algorithm tries to find a solution in every connected component.
Our final solution will be the one with the minimum cost among all connected component.
Note that for some connected component, there may not be a feasible solution
in that component (some target point can not be covered by any point in that component),
and our algorithm ignores such component.

From now on, we fix a connected component $\calC$ in $\Gc$. Let $\calG[\calC]$ be the collection of all edges in the connected component $\calC$.
Similar with Wu et al.~\cite{wu2013approximations},
we formulate the \mincsc\ problem as a group Steiner tree (\gst) problem.
Each edge $e\in \calG[\calC]$ is associated with a cost $c_e=1$.
For each target $p\in \calP$, we create a group
$$\group(p)=\calC\cap D(p, \Rs)=\{s \mid s\in \calC, \dist(p,s)\leq \Rs\}.$$
The goal is to find a tree $\calT$ (in $\calG[\calC]$) such that
$\calT\cap \group(p)\ne \emptyset$ for all $p\in \calP$
and the cost is minimized.
We can easily see the \gst\ instance constructed above is equivalent to
the original \mincsc\ problem
(the cost of the tree $\calT$ is the number of nodes in $\calT$ minus 1).
The \gst\ problem can be formulated as the following linear integral program:
We pick a root $r\in \calC$ for the tree $\calT$ and remove all target points that are covered by $r$ from $\calP$\footnote{We can do this since the final solution always contain $r$; see Equation~\eqref{eq:terminal}.}
(we need to enumerate all possible roots).
For each edge $e\in \calG[\calC]$, we use Boolean variable $x_e$
to denote whether we choose edge $e$.

\begin{align}
\label{eq:lpgst1}
\tag{$\mathsf{ILp}$\text{-}$\mathsf{GST}$}
\mathrm{minimize} \,\, & \sum_{e\in \calG[\calC]}x_e \\
\mathrm{subject\,\, to} \,\, & \sum_{e\in \partial(S)} x_e \geq 1,\quad\text{for all }  S\subset \calC \text{ such that } r\in S \text{ and }\exists p, S\cap G_p=\emptyset; \notag\\
& x_e\in \{0,1\}, \quad \forall e\in \calG[\calC] \notag.
\end{align}

The second constraint says that for any cut $\partial(S)$ that separates the root
$r$ from any group, there must be at least one chosen edge.
By replacing $x_e\in \{0,1\}$ with $x\in [0,1]$, we obtain the linear programming relaxation of \ref{eq:lpgst1} (denoted as $\lpr$).
By the duality between flow and cut, we can see that
the second constraint is equivalent to dictating that
we can send at least 1 unit of flow from the root $r$ to nodes in $\group(p)$, for each $p$.
This flow viewpoint (also observed in the original \gst\ paper~\cite{garg1998polylogarithmic})
will be particularly useful to us later.
So we write down the flow LP explicitly as follows.
We first replace every undirected edge $e=(u,v)$
by two directed arcs $(u,v)$ and $(v,u)$.
Let $\widehat{\calG}[\calC]$ denote the collection of all directed arcs.
For each $p\in \calP$ and
each directed arc $(u,v)$, we have a variable $x^p_{uv}$ indicating the flow of commodity $p$ on arc $(u,v)$.
We use $y^p_v= \sum_{u} x^p_{uv}-\sum_{w}x^p_{vw}$ to denote the net flow
(also called {\em flow excess}) of commodity $p$ into node $v$.
Then we develop the following linear program:

\begin{align}
\label{eq:lpgst}
\tag{$\lpflow$}
\mathrm{minimize} \,\, & \,\, \sum_{(u,v)\in \widehat{\calG}[\calC]} x_{uv} \\
\mathrm{subject\,\, to} \quad & y^p_v= \sum_{u:(u,v) \in \widehat{\calG}[\calC]} x^p_{uv}-\sum_{w:(v,w) \in \widehat{\calG}[\calC]}x^p_{vw} \quad \text{ for all }v\in \calC, p\in \calP\notag \\
& y^p_r=-1                                                                  \quad \text{ for all }p\in \calP, \notag \\
& \sum_{v\in \group(p)}y^p_v\geq 1                              \quad \text{ for all }p\in \calP, \notag \\
& y^p_u=0                                                                   \quad \text{ for all }p\in \calP, \text{ and } u\not\in \group(p), u\ne r,\notag \\
& x^p_{uv}\leq x_{uv}                               \quad\text{ for all } (u,v) \in \widehat{\calG}[\calC], p\in \calP, \notag \\
& x_{uv}, x^p_{uv}, y^p_v\in [0,1],                                        \quad \text{for all } (u,v) \in \widehat{\calG}[\calC], p\in \calP \notag.
\end{align}

We first have the following lemma that connects two programs \ref{eq:lpgst} and \ref{eq:lpgst1}.

\begin{lemma}
	\label{lm:optimal}
	The optimal value of \ref{eq:lpgst} is at most the optimal value of \ref{eq:lpgst1}.
\end{lemma}

\begin{proof}
	Given a feasible solution $\left\{x_e\right\}_{e\in \calG[\calC]}$, we construct a feasible solution of $\lpflow$ as follows: 
	\begin{enumerate}
		\item By definition, $\left\{x_e\right\}_{e\in \calG[\calC]}$ form a tree $T$ rooted at $r$.
		Denote $A\subseteq \widehat{\calG}[\calC]$ to be the collection of directed arcs satisfying that $u$ is the father point of $v$ on tree $T$. 
		\item For each directed arc $(u,v)\in A$, let $x_{uv}=1$.
		Otherwise, let $x_{uv}=0$.
		\item For each $p\in \calP$, there must exist a sensor $s_p\in \calC\cap G_p$ belonging to tree $T$ by the constraints of \ref{eq:lpgst1}.
		Denote $A_p\subseteq A$ to be the collection of directed arcs satisfying that both $u$ and $v$ lie on the unique path from root $r$ to $s_p$ on tree $T$. 
		\item For each $p\in \calP$ and each directed arc $(u,v)\in A_p$, let $x^p_{uv}=1$.
		Otherwise, let $x^p_{uv}=0$.
		\item For each $v\in \calC$ and $p\in \calP$, let $y^p_v= \sum_{u:(u,v) \in \widehat{\calG}[\calC]} x^p_{uv}-\sum_{w:(v,w) \in \widehat{\calG}[\calC]}x^p_{vw}$.
	\end{enumerate} 
	By construction, we can check that all constraints of $\lpflow$ are satisfied.
	Moreover, $\sum_{(u,v)\in \widehat{\calG}[\calC]} x_{uv} = \sum_{e\in \calG[\calC]}x_e$.
	This completes the proof.
	\qed
\end{proof}

Denote $\OPT$ to be the optimal fractional value of $\lpflow$.
Now, we describe our algorithm.
Our algorithm mainly consists of two steps.
In the first step, we extract a {\em geometric hitting set} instance from the optimal fractional solution of $\lpflow$.
We can find an integral solution $H$ for the hitting set problem and we can show its cost is at most $O(C^2\OPT)$.
Then by Lemma~\ref{lm:optimal}, the size of $H$ is at most $O(C^2)$ times the optimal value of \ref{eq:lpgst1}.
Moreover all sensors in $H$ can cover all target points $p\in \calP$.
In the second step, we extract a Steiner tree instance, again from the optimal fractional solution of $\lpflow$.
We show it is possible to round the Steiner tree LP to get a constant approximation integral Steiner tree,
which can connect all points in $H$.

\topic{Step 1: Constructing the Hitting Set Instance}

We first solve the linear program $\lpflow$ and obtain the fractional optimal solution
$(x_{uv}, y_v)$.
Let $\opt(\lpflow)$ to denote the optimal value of $\lpflow$.
We place a grid with grid size $l=\frac{\sqrt{2}}{2}$ in the plane (i.e., each cell is a $\frac{\sqrt{2}}{2}\times \frac{\sqrt{2}}{2}$ square).
W.l.o.g., we assume that grid lines are parallel to either the $x$-axis or the $y$-axis.
For each $p\in \calP$, consider the set of sensors $\group(p)$, that is the set of sensors which can cover $p$.
Since $\group(p)$ is contained in a disk $D(p, \Rs)$ of radius $\Rs\leq C$, the diameter of $D(p,\Rs)$ that is parallel to the $x$-axis is fully covered by at most $\lceil \frac{2\Rs}{l} \rceil\leq 2\sqrt{2} C +1$ grid cells.
Similarly, the diameter of $D(p,\Rs)$ that is parallel to the $y$-axis is also covered by at most $2\sqrt{2} C +1$ grid cells.
Thus, we conclude that there are at most
\rednote{$(2\sqrt{2} C +1)^2$} grid cells that may contain some points in $\group(p)$.
Since  $\sum_{v\in \group(p)}y^p_v\geq 1$, there must be a cell (say $\cell(p)$) such that
\begin{align}
\label{eq:cellsum}
\sum_{v\in \group(p)\cap \cell(p)} y^p_v\geq 1/(2\sqrt{2} C +1)^2\stackrel{C\geq 1}{\geq} 1/16C^2 =\Omega(1).
\end{align}
We call $\cell(p)$ the {\em significant cell} for point $p$. 
\footnote{
	If there are multiple such cells, we pick one arbitrarily.
}

Now, we construct a geometric hitting set (\hittingset) instance $(\calU, \calF)$ as follows:
Let the set of points be
$
\calU=\cup_{p\in \calP} (\group(p)\cap\cell(p))
$
and  the family of subsets be
$
\calF=\{\group(p)\}_{p\in \calP}.
$
The goal is to choose a subset $H$ of $\calU$ such that
$\group(p)\cap H\ne \emptyset$ for all $p\in \calP$ (i.e., we want to hit every set in $\calF$).
Write the linear program relaxation for the \hittingset\ problem (denoted as $\lphittingset$):

\begin{align}
\label{eq:lpgst}
\tag{$\lphittingset$}
\mathrm{minimize} \,\, & \,\, \sum_{u\in \calU} z_u \\
\mathrm{subject\,\, to} \quad & \sum_{u\in \group(p)\cap \cell(p)}z_u\geq 1                              \quad \text{ for all }p\in \calP, \notag \\
& z_u\in [0,1],                                        \quad \text{for all } u \in \calU \notag.
\end{align}
Let $\opt(\lphittingset)$ to denote the optimal value of $\lphittingset$.
We need the following simple lemma.

\begin{lemma}
	\label{lm:hittingset}
	$\opt(\lphittingset)\leq 16C^2 \opt(\lpflow)$.
\end{lemma}
\begin{proof}
	Suppose $(x_{uv}, y_v)$ is the optimal fractional solution for $\lpflow$.
	Now, we want to construct a feasible fractional solution $\{z_u\}_{u\in \calU}$ for $\lphittingset$
	such that $\sum_{u\in \calU} z_u\leq O(C^2 \sum_{uv}x_{uv})=O(C^2\opt(\lpflow))$.
	We simply let
	$$
	z_u=\min\{1, 16C^2  \max_{p\in \calP} y^p_u\}.
	$$
	From \eqref{eq:cellsum}, we can easily see $z_u$ is a feasible solution for the \hittingset\ problem:
	$$
	\sum_{u\in \group(p)\cap \cell(p)}z_u\geq
	\sum_{u\in \group(p)\cap \cell(p)}\min\{1, 16C^2 y^p_u\}\geq
	1                              \quad \text{ for all }p\in \calP
	$$
	It remains to see that
	\begin{align*}
	\sum_{u\in \calU}z_u & \leq \sum_{u\in \calU}16C^2 \max_p y^p_u \leq 16C^2 \sum_{u\in \calU} \max_{p\in \calP} \left(\sum_{w\in \calC} x^p_{wu}\right) \\
	& \leq 16C^2 \sum_{u\in \calU}  \sum_{w\in \calC} \max_{p\in \calP} \left(x^p_{wu}\right)  = 16C^2 \sum_{u\in \calU}  \sum_{w\in \calC} x_{wu} \\
	&\leq 16C^2 \sum_{uv}x_{uv}.
	\end{align*}
	This finishes the proof.
	\qed
\end{proof}


\rednote{C{\u{a}}linescu et al.~\cite{calinescu2004selecting} showed that we can round the above linear program $\lphittingset$ to obtain an integral solution (i.e., an actual hitting set) $H\subset \calU$ such that $|H|\leq 102\cdot \opt(\lphittingset)$.}\footnote{Note that $\lphittingset$ is equivalent to a minimum disk cover problem if we regard each $u\in \group(p)\cap \cell(p)$ as a unit disk of radius $\Rs$ centered at $u$. Hence, we can apply the rounding scheme for the minimum disk cover problem in~\cite{calinescu2004selecting}.}
In another work, Br{\"o}nnimann and Goodrich \cite{bronnimann1995almost},
combined with the existence of $\epsilon$-net of size $O(1/\epsilon)$ for disks (see e.g.,~\cite{pyrga2008new}), also showed that we can round $\lphittingset$
to an actual hitting set
$H\subset \calU$ such that $|H|\leq O(\opt(\lphittingset))$
(the connection to $\epsilon$-net was made simpler and more explicit in Even et al.~\cite{even2005hitting}).
Hence, by Lemma~\ref{lm:hittingset}, 
we have that $|H|\leq O(C^2\opt(\lpflow))$.

\topic{Step 2: Constructing the Steiner Tree Instance}
We now have a hitting set $H\subset \calU$. 
Consider a node $u\in H$.
Since $u$ is a node (a sensor) in the hitting set, 
we know there is some point $p_u\in \calP$ such that $u\in \group(p_u)\cap \cell(p_u)$. In other words, $u$ can cover $p_u$ and is in the significant
cell of $p_u$.
From \eqref{eq:cellsum}, we know that 
$\sum_{v\in \group(p_u)\cap \cell(p_u)} y^{p_u}_v\geq \Omega(1/C^2)$.

Consider the set of cells $\setcell=\{\cell(p_u)\mid u\in H\}$
\footnote{
	If a cell is the significant cell for more than one target point $p_u$,
	$\setcell$ only has one copy of the cell. In other words,
	it is indeed a {\em set} of cells.
}
If there is a cell which contains the root $r$, we exclude it from $\setcell$.
From each cell $\cell(p)\in \setcell$, we pick an arbitrary node (i.e., sensor) $v(\cell(p))$ in it, 
called the {\em representative node} of $\cell$.
By \ref{eq:cellsum} (i.e., $\sum_{v\in \group(p)\cap \cell(p)} y^p_v\geq \Omega(1/C^2)$), at least $\Omega(1/C^2)$ 
flow of commodity $p$ that enters $\cell(p)$.

Consider the Steiner tree problem in $\calG(\calC)$ 
in which the set of terminals is defined to be
\begin{equation}
\label{eq:terminal}
\terminal=\{r\}\cup \{v(\cell)\mid \cell\in \setcell\}.
\end{equation}
In another word, the goal of this Steiner tree problems
is to connect $r$ and all representative nodes.
We write down the following linear program relaxation for the
Steiner tree problem (denoted as $\lpsteiner$):
\begin{align}
\label{eq:lpsteiner}
\tag{$\lpsteiner$}
\mathrm{minimize} \,\, & \,\, \sum_{e\in \calG[\calC]} x_{e} \\
\mathrm{subject\,\, to} \,\, & \sum_{e\in \partial(S)} x_e \geq 1,\quad\text{for all } S\subset \calC \text{ such that } r\in S \text{ and }\exists \cell\in \setcell,  v(\cell)\not\in  S  \notag\\
& x_e\in [0,1], \quad \forall e\in \calG[\calC] \notag.
\end{align}

Now, we construct a feasible fractional solution for $\lpsteiner$
as follows.
Consider the optimal fractional solution $(x_{uv},y_v)$ of $\lpflow$.
We would like to construct another feasible fractional solution $(\hx_{uv},\hy_v)$ for $\lpflow$.
First, we construct an intermediate solution $(\tx_{uv},\ty_v)$ by
{\em rerouting} some flow. 
Then, we scale the flow to construct $(\hx_{uv},\hy_v)$.
The details are as follows:

\begin{itemize}
	\item (Flow Rerouting)
	Consider a cell $\cell(p)\in \setcell$.
	For each node $u\in \group(p)\cap\cell(p)$, let $\tx^p_{uv(\cell(p))}\leftarrow x^p_{uv(\cell(p))}+y^p_u$, and let $\tx^p_{uv}\leftarrow x^p_{uv}$ for any node $v\neq v(\cell(p))$. In other words, we route
	the flow excess at node $u$ to node $v(\cell(p))$.
	After such updates, for each $u\in \group(p)\cap\cell(p), u\ne v(\cell(p))$ we can see the flow excess is zero, or
	equivalently $\ty^p_u=0$.
	The flow excess at node $v(\cell(p))$ is
	\begin{equation}
	\label{eq:flowbound}
	\ty^p_{v(\cell(p))}=\sum_{v\in \group(p)\cap \cell(p)} y^p_v\geq 1/16C^2.
	\end{equation}
	We repeat the above process for all $\cell(p)\in \setcell$.
	\item
	We next increase the flow excess at node $v(\cell(p))$ to 1 for all $\cell(p)\in \setcell$, and construct another feasible solution $(\hx_{uv},\hy_v)$. For each $\cell(p)\in \setcell$, we define $(\hx^p_{uv},\hy^p_v)$ as follows: 
	\begin{enumerate}
		\item For each edge $(u,v)$, let $\hx^p_{uv}\leftarrow \tx^p_{uv}/\ty^p_{v(\cell(p))}$.
		Note that such scaling increases the flow excess at node $v(\cell(p))$ by a $1/\ty^p_{v(\cell(p))}$ factor. 
		\item For each node $v$, let $\hy^p_v\leftarrow \ty^p_{v}/\ty^p_{v(\cell(p))}$. 
	\end{enumerate}
	After the scaling, $1$ unit flow (thinking $\hx^p_{uv}$ as the flow value on $(u,v)$) enters $v(\cell(p))$ and $\hy^p_{v(\cell(p))}=1$. On the other hand, we have that $\hx^p_{uv}= \tx^p_{uv}/\ty^p_{v(\cell(p))} \leq 16C^2 \tx^p_{uv}$ for each edge $e$ following from the fact that $1/\ty^p_{v(\cell(p))}\leq 16C^2$. 
\end{itemize}

Let $\check{x}_{e}=\max_{p\in \calP} \hx^p_{uv}+\max_{p\in \calP}\hx^p_{vu}$, where $e$ is the undirected edge corresponding to directed edges $uv$ and $vu$ (Notice that $\lpflow$ is formulated on directed graphs and Steiner tree is formulated on undirected graphs.
).
For each $v(\cell(p))$, since at least 1 unit flow (thinking $\hx^p_{uv}$ as flow value on $(u,v)$) enters $v(\cell(p))$ and $\check{x}_{e}\geq \hx^p_{e}$, $\check{x}_{e}$ is a feasible solution for 
$\lpsteiner$.

Next, we show the optimal value of $\lpsteiner$
is not much larger than that of $\lpflow$.

\begin{lemma}
	\label{lm:steiner}
	$\opt(\lpsteiner)\leq O(C^2 \opt(\lpflow))$.
\end{lemma}
\begin{proof}
	Recall $\check{x}_{e}$ is a feasible solution for $\lpsteiner$ and $x_{uv}$ is the optimal solution for $\lpflow$. Also recall that $H\subset \calU$ is a hitting set instance satisfying that $|H|\leq O(C^2\opt(\lpflow))$.
	We only need to show that 
	$$
	\sum_{e\in \calG[\calC]} \check{x}_{e} \leq O(C^2\sum_{(u,v)\in \widehat{\calG}[\calC]} x_{uv} + |H|).
	$$
	This can be seen as follows:

	\begin{align*}
	\sum_{e\in \calG[\calC]} \check{x}_{e} = &\sum_{(u,v)\in \widehat{\calG}[\calC]} \max_{p\in \calP}\hx^p_{uv}
	= \sum_{(u,v)\in \widehat{\calG}[\calC]} \max_{p\in \calP}\tx^p_{uv}/\ty^p_{v(\cell(p))} \\
	\leq & \sum_{(u,v)\in \widehat{\calG}[\calC]}  \max_{p\in \calP} x^p_{uv}/ \ty^p_{v(\cell(p))} +\sum_{\cell(p) \in \setcell, u\in \group(p)\cap\cell(p)} y^p_{u}/\ty^p_{v(\cell(p))} \\
	\leq & \,\,\, 16C^2 \sum_{(u,v)\in \widehat{\calG}[\calC]}  \max_{p\in \calP} x^p_{uv} +\sum_{v(\cell(p))} 1 \\ 
	\leq &\,\,\, 16C^2 \sum_{(u,v)\in \widehat{\calG}[\calC]} x_{uv}+ |H|. 
	%
	\end{align*}
	The second equality follows from the construction of $\hx^p_{uv}$. The first inequality follows from the definition of $\tx^p_{uv}$ (we only reroute the flow for commodity $p$ such that $\cell(p)\in \setcell$, hence the second term). The second inequality follows from the fact that $\ty^p_{v(\cell(p))}\geq 1/16C^2$ and Equation~\eqref{eq:flowbound}.
	This finishes the proof of the lemma.
	\qed
\end{proof}

It is well known that the integrality gap of the Steiner tree problem is a constant \cite{williamson2011design}.
In particular, it is known that using the primal-dual method (based on $\lpsteiner$) in~\cite{goemans1995general}
(see also \cite[Chapter 7.2]{williamson2011design}),
we can obtain an integral solution $\bx_e$ such that
$$
\sum_{e\in \calG[\calC]} \bx_e \leq 2\opt(\lpsteiner) \leq O(C^2 \opt(\lpflow)) \leq O(C^2\OPT).
$$
Let $J$ be the set of vertices spanned by the integral Steiner tree $\{\bx_e\}$.
The above discussion shows that $|J|\leq O(C^2\OPT)$.
Our final solution (the set of sensors we choose) is
$\sol=H\cup J.$
The feasibility of $\sol$ is proved in the following simple lemma.
\begin{lemma}
	$\sol$ is a feasible solution.
	\label{lm:final1}
\end{lemma}
\begin{proof}
	We only need to show that $\sol$
	induces a connected graph and covers all the target points.
	Obviously, $H$ covers all target points, so does $\sol$.
	Since $J$ is a Steiner tree, thus connected. Moreover, $J$ connects
	all representatives $v(\cell)$ for all $\cell\in \setcell$. On the other hand,
	$H$ only contains those sensors in $\cell\in \setcell$.
	So every sensor in $v\in H$ (say $v\in\cell$) is connected to the representative $v(\cell)$.
	So $H\cup J$ induces a connected subgraph.
	\qed
\end{proof}

Lastly, we need to show the performance guarantee.
This is easy since we have shown that
both $|H|\leq O(C^2\OPT)$ and $|J|\leq O(C^2\OPT)$.
So $|\sol|=O(C^2\OPT)=O(\OPT)$ since $C$ is assumed to be a constant.

\section{Budgeted Connected Sensor Cover}
\label{sec:bcsc}

Again we assume that $\Rc=1$ and $\Rs=C$.
Recall that our goal is to find a subset $\calS'\subseteq \calS$
of sensors with cardinality $B$ which induces a connected subgraph and
covers as many targets as possible.
We first construct the communication graph $\Gc$ as in Section~\ref{sec:mscs}.
Again, we only need to focus on a connected component of $\Gc$.
Then we find a square $Q$ in the Euclidean plane large enough such that all of the $n$ sensors are inside $Q$.
Similar to \cite{marathe1995simple,hunt1998nc}, we partition $Q$ into small square cells of equal size.
Let the side length of each cell be $l=\frac{\sqrt{2}}{2}$.
Denote the cell in the ith row and jth column of the partition as $\cell_{i,j}$.
Let $V_{i,j} = \{v\in \calS \mid v\in \cell_{i,j} \}$ be the collection of sensors in $\cell_{i,j}$.
We then partition these cells into $k^2$ different cell groups
$\cells_{a,b}$,
where $k = \lceil 2C/l + 1 \rceil$.
In particular, we let
$$
\cells_{a,b} = \{\cell_{i,j}\mid i\equiv a (\mod k), j\equiv b (\mod  k)\} \text{ for }a\in [k], b \in [k],
$$
and
$\calV_{a,b} = \calS \cap \cells_{a,b}$ be the collection of sensors in $\cells_{a,b}$; see Figure~\ref{fig:grid} as an example.
\begin{figure}
	\label{fig:grid}
	\caption{Partition cells into $2^2$ different cell groups $\calV_{1,1}, \calV_{1,0},\calV_{0,1},\calV_{0,0}$.}
	\centering
	\includegraphics[height=130pt]{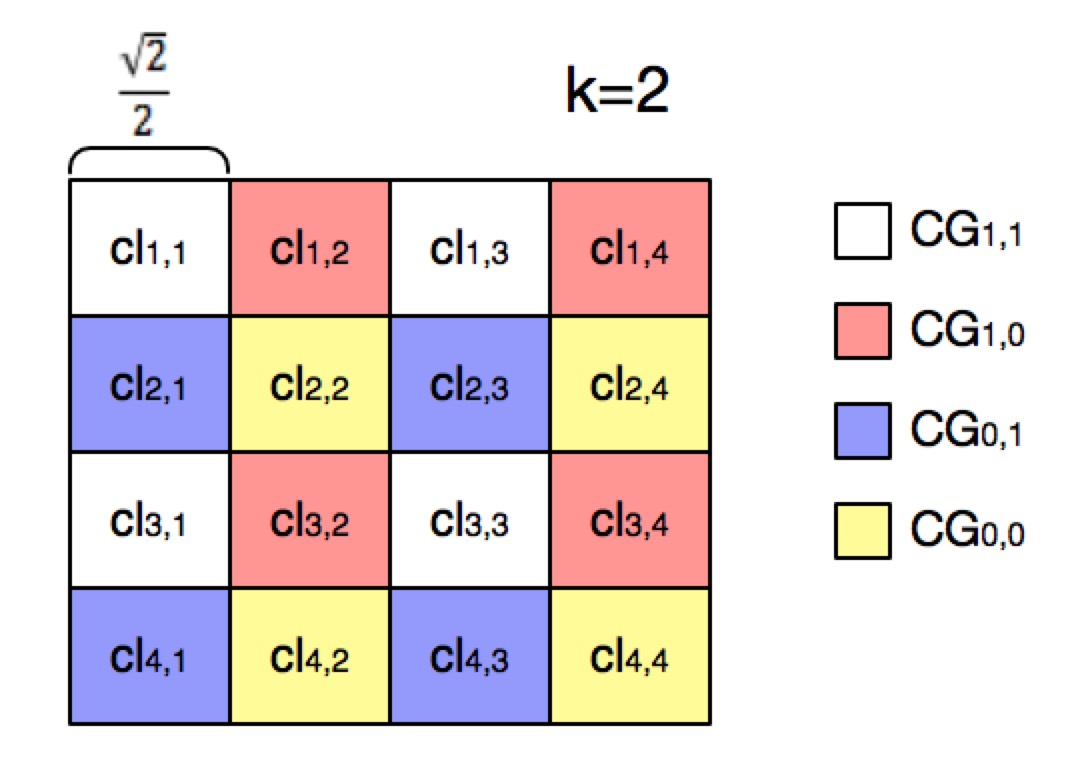}
\end{figure}

With the above value $k$, we make a simple but useful observation as follows.
\begin{observation}
	\label{ob:nodoublecover}
	There is no target covered by two different sensors contained in two different cells of $\cells_{a,b}$.
\end{observation}
Denote the optimal solution of \bcsc\ problem as $\OPT$.
In this section, we present an $O\bigl(\frac{1}{C^2}\bigr)$ factor approximation algorithm
for the \bcsc\ problem.

\subsection{The Algorithm}
For $0 \leq a , b < k$, we repeat the following two steps,
and output a tree $T$ with $O\bigl(B\bigr)$ vertices (sensors) which covers the maximum number of targets.
Then based on $T$, we find a subtree $\tilde{T}$ with exactly $B$ vertices as our final output.

\begin{algorithm}[t]
	\label{alg_GSC}
	\caption{Reassign profits via the greedy algorithm}
	\textbf{Input:} The sensor collection $\calS$, the target collection $\calP$, the cell collection $\cells_{a,b}$. \\
	\textbf{Output:} Profit function $\hp: \calP \to \mathbb{Z}^+ \cup \{0\}$\begin{enumerate}
		\item \textbf{for all} $\cell_{i,j}\in \cells_{a,b}$ \textbf{do}
		{\setlength\itemindent{15pt} \item $P_t \gets \calP$ ~~~~~~ //$P_t$ is the set of uncovered targets}
		{\setlength\itemindent{15pt} \item $V_s \gets V_{i,j}$  ~~~~~~ //$V_s$ is the set of available sensors\\
			{\setlength\itemindent{15pt} \item \textbf{for all} $v \in \calS$ \textbf{do} }
			\begin{enumerate}
				{\setlength\itemindent{15pt} \item $v \gets \mathop{\arg\max}_{v \in P_t} |N_{P_t}(v)|$  ~~~~~
					//$N_{P_t}(v)$ is the set of uncovered targets that can be covered by $v$. 		}
				{\setlength\itemindent{15pt} \item $\hp(v) \gets |N_{P_t}(v)|$,
					$P_t \gets P_t \backslash N_{P_t}(v)$,
					$V_s \gets V_s \backslash \{v\}$ }		
			\end{enumerate}
		}
		{\setlength\itemindent{15pt} \item \textbf{end for} }
		\item \textbf{end for}
		\item return $\hp$
	\end{enumerate}
\end{algorithm}

\topic{Step 1: Reassign profit}
The profit $p(S)$ of a subset $S\subseteq \calS$ is the number of targets covered by $S$.
$p(S)$ is a submodular function.
In this step,
we design a new profit function (called {\em modified profit function}) $\hp: \calS\rightarrow\Z^+$ for the set of sensors.
To some extent, $\hp$ is a linearized version of $p$ (module a constant approximation factor).

Now, we explain in details how $\hp$ is defined.
Fix a cell group $\cells_{a,b}$.
\footnote{
	For each $\cells_{a,b}$, we define a modified profit function $\hp_{a,b}$.
	For ease of notation, we omit the subscripts.
}
For the vertices in $\calV_{a,b}$, we use the greedy algorithm Algorithm \ref{alg_GSC}
to reassign profits of the vertices in $\calV_{a,b}$. Generally speaking, we greedily pick a vertex which covers the most number of targets each time, and use this number as the modified profit. The details are as follows.
Among all vertices in $\calV_{a,b}$,
we pick a vertex $v_1$ which can cover the most number of targets,
and use this number as its modified profit $\hp(v_1)$.
Remove the chosen vertex and targets covered by it.
We continue to pick the vertex $v_2$ in $\calV_{a,b}$ which can cover the most number of uncovered targets.
Set the modified profit $\hp(v_2)$ to be the number of newly covered targets.
Repeat the above steps until all the sensors in $\calV_{a,b}$ have been picked out.
For other vertices $v$ which are not in $\calV_{a,b}$, we simply set their modified profit $\hp(v)$ as 0.

Let us first make some simple observations about $p$ and $\hp$.
We use $\hp(S)$ to denote $\sum_{v\in S}\hp(v)$.
First, it is not difficult to see that $\hp(S)\leq p(S)$ for any subset $S\subseteq\calS$.
Second, we can see that it is equivalent to run the greedy algorithm for
each cell in $\cells_{a,b}$ separately (due to Observation~\ref{ob:nodoublecover}).
Suppose $S_1\subseteq \cell_{c,d}$,
$S_2\subseteq \cell_{c',d'}$ where $\cell_{c,d}$ and $\cell_{c',d'}$ are two different cells in $\cells_{a,b}$,
then $p(S_1\cup S_2)=p(S_1)+p(S_2)$ due to Observation~\ref{ob:nodoublecover}.

Consider a cell $\cell_{c,d}\in \cells_{a,b}$.
Let $D_{c,d} = \{v_1, v_2, ... ,v_n\}\subseteq \cell_{c,d}\cap \calS$,
where the vertices are indexed by the order
in which they were selected by the greedy algorithm.
Let $D_{c,d}^{i} = \{v_1, v_2, ... , v_i\}$ be the first $i$ vertices in $D_{c,d}$.
By the following lemma, we can see that the modified profit function $\hp$ is a constant approximation to
true profit function $p$ over any vertex subset $V\subseteq \calV_{a,b}$.

\begin{lemma}
	For a set of vertices $V$ in the same cell $\cell_{c,d}\in \cells_{a,b}$, such that $|V|\leq i$,
	we have that
	$p(D_{c,d}^{i})=\hp(D_{c,d}^{i}) \geq (1-1/e) p(V)$.
	\label{lm_1/e}
\end{lemma}
\begin{proof}
	\eat{
		Suppose in Algorithm~\ref{alg_GSC}, we select
		the vertices in $V$ in the order of $i_1,i_2,\ldots,i_{|V|}$.
		Let $V_j=\{v_{i_1},v_{i_2},\ldots,v_{i_j}\}$ be the first $j$ vertices selected by the greedy algorithm in $V$.
		By the construction of $\hp$, we have
		$$\hp(V)=\sum_{1\leq j\leq |V|}\bigl( p(D_{i_j})-p(D_{i_{j}-1})\bigr) $$
		Note that $p(V)=\sum_{1\leq j\leq |V|}\bigl(p(\{V_{j}\})-p(V_{j-1})\bigr)$. Since $p$ is a submodule function and $V_j\subseteq D_{i_j}$ for $1\leq j\leq |V|$. So we have $p(V) \geq \hp(V)$.
	}
	By the greedy rule, we can see $p(D_{c,d}^{i})=\hp(D_{c,d}^{i})$.
	By Lemma \ref{lm_bsc}, we know that $\hp(D_{c,d}^{i}) \geq  (1-1/e) \max_{|V| \leq i} p(V)$.
	\qed
\end{proof}

\topic{Step 2: Guess the optimal profit and calculate a tree $T$} Although the actual profit of $\OPT$ is unknown, we can guess the profit of $\OPT$ (by enumerating all possibilities).
For each $0\leq a,b<k$, we calculate in this step
a tree $T$ of size at most $4B$, using the \qst\ algorithm (see Lemma~\ref{lm:qst}).
We can show that among these trees (for different $a,b$ values),
there must be one tree of profit no less than $\frac{1}{k^2}\left(1-\frac{1}{e}\right)\OPT$.

After choosing the best tree $T$ with the highest profit,
we construct a subtree $\tilde{T}$ of size $B$ based on $T$ as our final solution of \bcsc.

\begin{algorithm}
	\label{alg_gppafb}
	\caption{Algorithm for \bcsc\ with greedy profit assignment}
	\textbf{Input:} The sensor collection $\calS$, the target collection $\calP$, budget B. \\
	\textbf{Output:} a tree $\tilde{T}$ with $|\tilde{T}|\leq B $.
	\begin{enumerate}
		\item Construct the communication graph $\Gc$	
		\item \textbf{for} $a$ \textbf{from} $0$ to $k-1$, $b$ \textbf{from} $0$ to $k-1$
		\begin{enumerate}
			\item Reassign every vertex's profit with Algorithm \ref{alg_GSC} \rednote{and obtain a profit function $\hat{p}$}.
			\item Set every edge's cost as 1
			\item $ProfitOpt_{guess} \gets 1$
			\item \textbf{Do}
			\begin{enumerate}
				\item $T' \gets$ Run the $2$-approximation algorithm of \qst\ on $\Gc$ \rednote{with the profit function $\hat{p}$} and quota $ProfitOpt_{guess}$
				\item \textbf{if} $|T'| \leq 4B$ \textbf{then} $T \gets T'$
				\item $ProfitOpt_{guess} = ProfitOpt_{guess} + 1$
			\end{enumerate}
			\item \textbf{While}$(|T'| \leq 4B)$		
		\end{enumerate}			
		\item \textbf{end for}
		\item $\tilde{T} \gets $ use the dynamic programming algorithm described in Section 5.2.2 in \cite{khuller2014analyzing} to find the best profit subtree of size $B$ from $T$.
		\item return $\tilde{T}$	
	\end{enumerate}
\end{algorithm}

We first show that there exists $0\leq a,b <k$, such that based on the modified profit $\hp$ on $\cells_{a,b}$,
there exists a tree with at most $2B$ vertices of total modified profit at least
$\frac{1}{k^2}\left(1-\frac{1}{e}\right)\OPT$.
We use $T_{\OPT}$ to denote the set of vertices of the optimal solution.

\begin{lemma}
	There exists a tree $T_0$ in $\Gc$,
	$|T_0| \leq 2B$ such that $\displaystyle  \hp(T_0) \geq \frac{1}{k^2}\left(1-\frac{1}{e}\right) \OPT $
	\label{lm_1/k^2}
\end{lemma}
\begin{proof}
	We first notice that
	\begin{align*}
	\OPT = p\left(\displaystyle \bigcup_{0\leq a , b <k} T_{\OPT} \cap \cells_{a,b}\right)
	\leq \displaystyle \sum_{0\leq a ,b <k} p\left(T_{\OPT} \cap \cells_{a,b}\right).
	\end{align*}
	Hence, there exists $0 \leq a',b' < k$, such that
	\begin{align*}
	p(T_{\OPT} \cap \cells_{a',b'}) \geq \frac{1}{k^2} \displaystyle \sum_{0\leq a , b <k} p(T_{\OPT} \cap \cells_{a,b})
	\geq \frac{1}{k^2} \OPT.
	\end{align*}
	For any cell $\cell_{c,d} \in \cells_{a',b'}$, suppose $n_{c,d}=|T_{\OPT}\cap \cell_{c,d}|$.
	$T_0$ is obtained from
	$T_{\OPT}$ by appending all vertices in $D_{c,d}^{n_{c,d}}$ (recall that $D_{c,d}^{n_{c,d}}$ consists of
	the first $n_{c,d}$ vertices selected in $\cell_{c,d}$ by the greedy algorithm).
	Note that we append at most $B$ vertices in total, and all vertices are still connected
	(
	since all vertices in the same cell are connected
	).
	Thus, $T_0$ is connected and has at most $2B$ vertices.
	
	By Lemma \ref{lm_1/e}, we can see that $\hp(D_{c,d}^{n_{c,d}}) \geq \left(1-\frac{1}{e}\right)p(T_{\OPT} \cap \cell_{c,d})$. Thus, we have \begin{align*}
	\hp(T_0)&=\sum_{\cell_{c,d}\in\cells_{a,b}}\hp(D_{c,d}^{n_{c,d}}) \geq \left(1-\frac{1}{e}\right)\sum_{\cell_{c,d}\in\cells_{a,b}}p\left(T_{\OPT} \cap \cell_{c,d}\right)\\
	&=\left(1-\frac{1}{e}\right)p(T_{\OPT} \cap \cells_{a,b}) \geq \frac{1}{k^2}\left(1-\frac{1}{e}\right)\OPT.
	\end{align*}
	Both equalities hold due to Observation~\ref{ob:nodoublecover}.
	\qed
\end{proof}

Then, by Lemma~\ref{lm:qst} and Lemma~\ref{lm_1/k^2},
if we run the \qst\ algorithm (with $\hp$ as the profit function),
we can obtain the suitable tree $T$ with at most $4B$ vertices of profit at least $\frac{1}{k^2}\left(1-\frac{1}{e}\right)p(\OPT)$.
The pseudocode of the algorithm can be found in Algorithm~\ref{alg_gppafb}.

\begin{lemma}
	Let $T$ be the tree obtained in Algorithm~\ref{alg_gppafb}, then
	$p(T) \geq \frac{1}{k^2}\left(1-\frac{1}{e}\right)\OPT$
	\label{lm_treeAppro}
\end{lemma}
\begin{proof}
	By Lemma~\ref{lm_1/k^2}, we can obtain a tree $T$ with at most $4B$ nodes. We also have $\hp(T)\geq \frac{1}{k^2}\left(1-\frac{1}{e}\right)\OPT$.
	Since $p(S)\geq \hp(S)$ for any $S$, we have that
	$
	p(T) \geq \frac{1}{k^2}\left(1-\frac{1}{e}\right)\OPT.
	$
	\qed
\end{proof}

Then we show how to construct a subtree $\tilde{T}$ of $B$ vertices based on tree $T$. Our technique is the same as Khuller et al. \cite{khuller2014analyzing}.
Firstly, they use the following theorem by Jordan \cite{jordan1869assemblages} to prove Lemma \ref{lm_khuller}. Then by a carefully partition, they obtain a subtree with $B$ vertices of profit at least $\frac{1}{13}$ of original tree with $6B$ vertices. Our construction is almost the same except that the original tree $T$ in our setting has at most $4B$ vertices.

\begin{lemma}[Jordan \cite{jordan1869assemblages}]
	Given any tree on n vertices, we can decompose it into two trees (by replicating a single vertex) such that the smaller tree has at most $\lceil \frac{n}{2} \rceil$ nodes and the larger tree has at most $\lceil \frac{2n}{3} \rceil$ nodes.
	\label{lm_jordan}
\end{lemma}
\begin{lemma}[Khuller et al. \cite{khuller2014analyzing}]
	Let $B$ be greater than a sufficiently large constant. Given a tree $T$ with $6B$ nodes, we can \rednote{partition the vertex set of $T$} it into 13 trees of size at most $B$ nodes each.
	\label{lm_khuller}
\end{lemma}

Denote the subtree with highest total profit as $\tilde{T}$. 
By the above lemma, $\tilde{T}$ has at most $B$ nodes.
Then we show the following lemma.
\begin{lemma}
	Assume $B\geq 10$. $p(\tilde{T}) \geq \frac{1}{8} p(T)$
	\label{th_treeDecompose}
\end{lemma}
\begin{proof} By Lemma \ref{lm_jordan}, we decompose the tree $T$ into two trees $T_1$ and $T_2$ such that $|T_1| \leq 2B $ and $|T_2| \leq \frac{8}{3}B+1$ and continue decomposing until the tree has at most $k$ vertices (as shown in the figure. Note that each subtree in the white square in the figure has at most $B$ vertices).
	Thus we can decompose a tree of size $4B$ to at most 8 subtrees of size at most $B$.
	See the figure.
	Suppose the subtrees are $T_1$,$T_2$,...,$T_8$. Then we have,
	\begin{align*}
	p(\tilde{T}) \geq \frac{1}{8} \sum_{i=1}^{8} p(T_i) \geq \frac{1}{8} p(T)
	\end{align*}
	So there is a subtree of size at most $k$ and profit at least $\frac{1}{8}p(T)$.
	\qed
\end{proof}
\includegraphics[height=130pt]{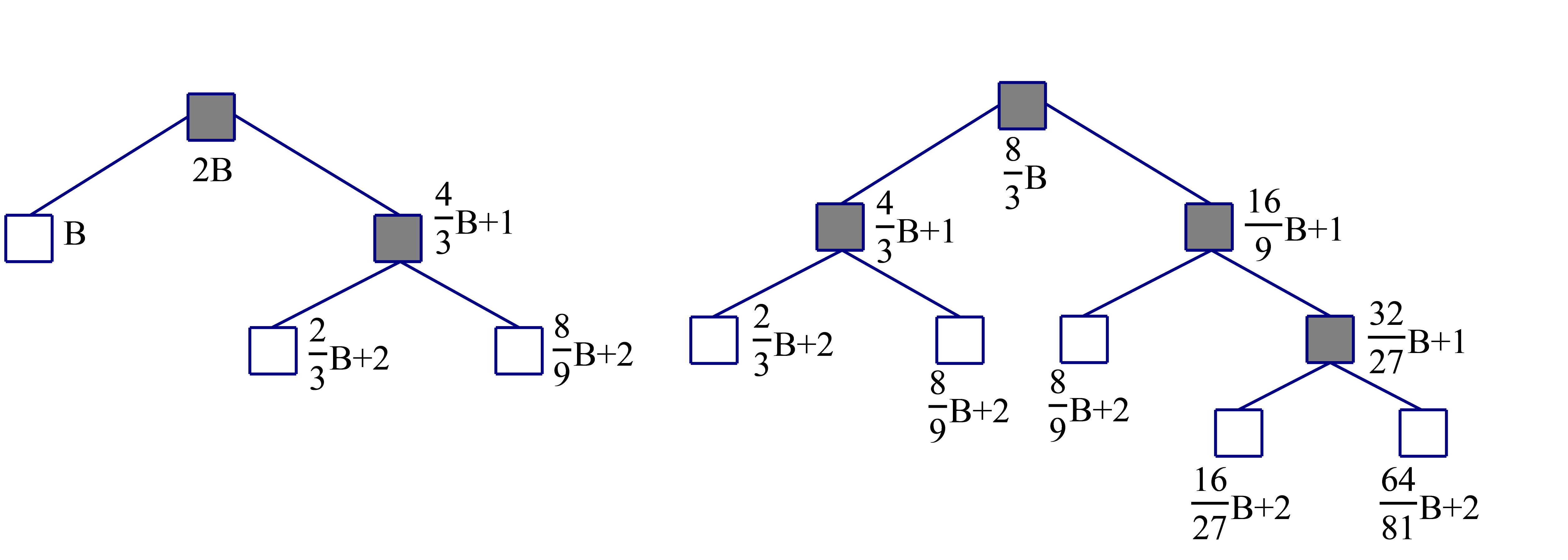}

Use the same dynamic programming algorithm in Khuller et al.~\cite{khuller2014analyzing},
we can find $\tilde{T}$ from tree $T$.
Combining Lemma \ref{lm_treeAppro} and Lemma \ref{th_treeDecompose},
$
p(\tilde{T}) \geq \frac{1}{8} \left(1-\frac{1}{e}\right) \frac{1}{(2\sqrt{2}C + 1)^2}\OPT = \frac{1}{12.66(8C^2+4\sqrt{2}C+1)}\OPT \geq \frac{1}{102C^2}\OPT$ (if $C\geq 100$).

Thus, we have obtained Theorem \ref{th_bcsc}.


\section{Conclusion and Future Work}
There are several interesting future directions.
The first obvious open question is that whether we
can get constant approximations for \mincsc\ and \bcsc\ without Assumption~\ref{ass:1}
(it would be also interesting to obtain approximation ratios that have better dependency on $C$).
Generalizing the problem further, an interesting future direction is the case where
different sensors have different transmission ranges and sensing ranges.
Whether the problems admit better approximation ratios than the (more general) graph theoretic
counterparts is still wide open.
Another interesting future direction is to obtain constant approximations for the weighted versions of  \mincsc\ and \bcsc.

\section{Acknowledgments}

We would like to thank anonymous reviewers for their constructive comments,
and pointing out a problematic argument in a previous version of the paper.
We also would like thank Dingzhu Du and Zhao Zhang for helpful discussions.
The research is supported in part by the National Basic Research Program of China Grant 2015CB358700, 
the National Natural Science Foundation of China Grant 61822203, 61772297, 61632016, 61761146003,
and a grant from Microsoft Research Asia.

\bibliographystyle{plain}
\bibliography{sensor}

\begin{thebibliography}{10}

\bibitem{ambuhl2006constant}
Christoph Amb{\"u}hl, Thomas Erlebach, Mat{\'u}{\v{s}} Mihal{\'a}k, and Marc
  Nunkesser.
\newblock Constant-factor approximation for minimum-weight (connected)
  dominating sets in unit disk graphs.
\newblock In {\em Approximation, Randomization, and Combinatorial Optimization.
  Algorithms and Techniques}, pages 3--14. Springer, 2006.

\bibitem{bronnimann1995almost}
Herv{\'e} Br{\"o}nnimann and Michael~T Goodrich.
\newblock Almost optimal set covers in finite {VC}-dimension.
\newblock {\em Discrete \& Computational Geometry}, 14(1):463--479, 1995.

\bibitem{calinescu2011maximizing}
Gruia C{\u{a}}linescu, Chandra Chekuri, Martin P{\'a}l, and Jan Vondr{\'a}k.
\newblock Maximizing a monotone submodular function subject to a matroid
  constraint.
\newblock {\em SIAM Journal on Computing}, 40(6):1740--1766, 2011.

\bibitem{calinescu2004selecting}
Gruia C{\u{a}}linescu, Ion~I. Mandoiu, Peng{-}Jun Wan, and Alexander
  Zelikovsky.
\newblock Selecting forwarding neighbors in wireless ad hoc networks.
\newblock {\em {MONET}}, 9(2):101--111, 2004.

\bibitem{calinescu2005polymatroid}
Gruia C{\u{a}}linescu and Alexander Zelikovsky.
\newblock The polymatroid {S}teiner problems.
\newblock {\em Journal of Combinatorial Optimization}, 9(3):281--294, 2005.

\bibitem{Chan2012}
Timothy~M. Chan, Elyot Grant, Jochen K\"{o}nemann, and Malcolm Sharpe.
\newblock Weighted capacitated, priority, and geometric set cover via improved
  quasi-uniform sampling.
\newblock In {\em Proceedings of the Twenty-third Annual ACM-SIAM Symposium on
  Discrete Algorithms}, SODA '12, pages 1576--1585. SIAM, 2012.

\bibitem{chekuri2006greedy}
Chandra Chekuri, Guy Even, and Guy Kortsarz.
\newblock A greedy approximation algorithm for the group {S}teiner problem.
\newblock {\em Discrete Applied Mathematics}, 154(1):15--34, 2006.

\bibitem{cheng2003polynomial}
Xiuzhen Cheng, Xiao Huang, Deying Li, Weili Wu, and Ding-Zhu Du.
\newblock A polynomial-time approximation scheme for the minimum-connected
  dominating set in ad hoc wireless networks.
\newblock {\em Networks}, 42(4):202--208, 2003.

\bibitem{Clarkson}
Kenneth~L. Clarkson and Kasturi Varadarajan.
\newblock Improved approximation algorithms for geometric set cover.
\newblock {\em Discrete \& Computational Geometry}, 37(1):43--58, 2007.

\bibitem{dai20095+}
Decheng Dai and Changyuan Yu.
\newblock A 5+ $\epsilon$-approximation algorithm for minimum weighted
  dominating set in unit disk graph.
\newblock {\em Theoretical Computer Science}, 410(8):756--765, 2009.

\bibitem{dinur2014analytical}
Irit Dinur and David Steurer.
\newblock Analytical approach to parallel repetition.
\newblock In {\em Proceedings of the 46th Annual ACM Symposium on Theory of
  Computing}, pages 624--633. ACM, 2014.

\bibitem{du2012connected}
Ding-Zhu Du and Peng-Jun Wan.
\newblock {\em Connected Dominating Set: Theory and Applications}, volume~77.
\newblock Springer Science \& Business Media, 2012.

\bibitem{elbassioni2012relation}
Khaled~M. Elbassioni, Slobodan Jelic, and Domagoj Matijevic.
\newblock The relation of connected set cover and group {S}teiner tree.
\newblock {\em Theor. Comput. Sci.}, 438:96--101, 2012.

\bibitem{even2005hitting}
Guy Even, Dror Rawitz, and Shimon~Moni Shahar.
\newblock Hitting sets when the {VC}-dimension is small.
\newblock {\em Information Processing Letters}, 95(2):358--362, 2005.

\bibitem{fakcharoenphol2003tight}
Jittat Fakcharoenphol, Satish Rao, and Kunal Talwar.
\newblock A tight bound on approximating arbitrary metrics by tree metrics.
\newblock In {\em Proceedings of the thirty-fifth annual ACM symposium on
  Theory of computing}, pages 448--455. ACM, 2003.

\bibitem{feige1998threshold}
Uriel Feige.
\newblock A threshold of $\ln n$ for approximating set cover.
\newblock {\em Journal of the ACM (JACM)}, 45(4):634--652, 1998.

\bibitem{funke2007improved}
Stefan Funke, Alexander Kesselman, Fabian Kuhn, Zvi Lotker, and Michael Segal.
\newblock Improved approximation algorithms for connected sensor cover.
\newblock {\em Wireless Networks}, 13(2):153--164, 2007.

\bibitem{garg2005saving}
Naveen Garg.
\newblock Saving an epsilon: a 2-approximation for the $k$-{MST} problem in
  graphs.
\newblock In {\em Proceedings of the thirty-seventh annual ACM symposium on
  Theory of computing}, pages 396--402. ACM, 2005.

\bibitem{garg1998polylogarithmic}
Naveen Garg, Goran Konjevod, and R~Ravi.
\newblock A polylogarithmic approximation algorithm for the group {S}teiner
  tree problem.
\newblock In {\em Proceedings of the ninth annual ACM-SIAM symposium on
  Discrete algorithms}, pages 253--259. Society for Industrial and Applied
  Mathematics, 1998.

\bibitem{goemans1995general}
Michel~X Goemans and David~P Williamson.
\newblock A general approximation technique for constrained forest problems.
\newblock {\em SIAM Journal on Computing}, 24(2):296--317, 1995.

\bibitem{guha1999improved}
Sudipto Guha and Samir Khuller.
\newblock Improved methods for approximating node weighted {S}teiner trees and
  connected dominating sets.
\newblock {\em Information and computation}, 150(1):57--74, 1999.

\bibitem{gupta2006connected}
Himanshu Gupta, Zongheng Zhou, Samir~R Das, and Quinyi Gu.
\newblock Connected sensor cover: self-organization of sensor networks for
  efficient query execution.
\newblock {\em Networking, IEEE/ACM Transactions on}, 14(1):55--67, 2006.

\bibitem{hochbaum1998analysis}
Dorit~S Hochbaum and Anu Pathria.
\newblock Analysis of the greedy approach in problems of maximum $k$-coverage.
\newblock {\em Naval Research Logistics}, 45(6):615--627, 1998.

\bibitem{huang2009better}
Yaochun Huang, Xiaofeng Gao, Zhao Zhang, and Weili Wu.
\newblock A better constant-factor approximation for weighted dominating set in
  unit disk graph.
\newblock {\em Journal of Combinatorial Optimization}, 18(2):179--194, 2009.

\bibitem{hunt1998nc}
Harry~B Hunt~III, Madhav~V Marathe, Venkatesh Radhakrishnan, Shankar~S Ravi,
  Daniel~J Rosenkrantz, and Richard~E Stearns.
\newblock {NC}-approximation schemes for {NP}-and {PSPACE}-hard problems for
  geometric graphs.
\newblock {\em Journal of algorithms}, 26(2):238--274, 1998.

\bibitem{johnson2000prize}
David~S. Johnson, Maria Minkoff, and Steven Phillips.
\newblock The prize collecting steiner tree problem: theory and practice.
\newblock In {\em Proceedings of the Eleventh Annual {ACM-SIAM} Symposium on
  Discrete Algorithms, January 9-11, 2000, San Francisco, CA, {USA.}}, pages
  760--769, 2000.

\bibitem{jordan1869assemblages}
Camille Jordan.
\newblock Sur les assemblages de lignes.
\newblock {\em J. Reine Angew. Math}, 70(185):81, 1869.

\bibitem{khuller2014analyzing}
Samir Khuller, Manish Purohit, and Kanthi~K Sarpatwar.
\newblock Analyzing the optimal neighborhood: algorithms for budgeted and
  partial connected dominating set problems.
\newblock In {\em Proceedings of the Twenty-Fifth Annual ACM-SIAM Symposium on
  Discrete Algorithms}, pages 1702--1713. SIAM, 2014.

\bibitem{klein1995nearly}
Philip Klein and R~Ravi.
\newblock A nearly best-possible approximation algorithm for node-weighted
  {S}teiner trees.
\newblock {\em Journal of Algorithms}, 19(1):104--115, 1995.

\bibitem{kuo2013maximizing}
Tung-Wei Kuo, KC-J Lin, and Ming-Jer Tsai.
\newblock Maximizing submodular set function with connectivity constraint:
  Theory and application to networks.
\newblock In {\em INFOCOM, 2013 Proceedings IEEE}, pages 1977--1985. IEEE,
  2013.

\bibitem{li2015ptas}
Jian Li and Yifei Jin.
\newblock A {PTAS} for the weighted unit disk cover problem.
\newblock In {\em International Colloquium on Automata, Languages, and
  Programming}, pages 898--909. Springer, 2015.

\bibitem{lichtenstein1982planar}
David Lichtenstein.
\newblock Planar formulae and their uses.
\newblock {\em SIAM journal on computing}, 11(2):329--343, 1982.

\bibitem{marathe1995simple}
Madhav~V. Marathe, H.~Breu, Harry B.~Hunt III, S.~S. Ravi, and Daniel~J.
  Rosenkrantz.
\newblock Simple heuristics for unit disk graphs.
\newblock {\em Networks}, 25(2):59--68, 1995.

\bibitem{mustafa2009ptas}
Nabil~Hassan Mustafa and Saurabh Ray.
\newblock {PTAS} for geometric hitting set problems via local search.
\newblock In {\em Proceedings of the twenty-fifth annual symposium on
  Computational geometry}, pages 17--22. ACM, 2009.

\bibitem{nemhauser1978analysis}
George~L Nemhauser, Laurence~A Wolsey, and Marshall~L Fisher.
\newblock An analysis of approximations for maximizing submodular set
  functions.
\newblock {\em Mathematical Programming}, 14(1):265--294, 1978.

\bibitem{pyrga2008new}
Evangelia Pyrga and Saurabh Ray.
\newblock New existence proofs $\varepsilon$-nets.
\newblock In {\em Proceedings of the twenty-fourth annual symposium on
  Computational geometry}, pages 199--207. ACM, 2008.

\bibitem{varadarajan2010weighted}
Kasturi Varadarajan.
\newblock Weighted geometric set cover via quasi-uniform sampling.
\newblock In {\em Proceedings of the forty-second ACM symposium on Theory of
  computing}, pages 641--648. ACM, 2010.

\bibitem{vondrak2011submodular}
Jan Vondr{\'a}k, Chandra Chekuri, and Rico Zenklusen.
\newblock Submodular function maximization via the multilinear relaxation and
  contention resolution schemes.
\newblock In {\em Proceedings of the forty-third annual ACM symposium on Theory
  of computing}, pages 783--792. ACM, 2011.

\bibitem{wan2002distributed}
Peng-Jun Wan, Khaled~M Alzoubi, and Ophir Frieder.
\newblock Distributed construction of connected dominating set in wireless ad
  hoc networks.
\newblock In {\em INFOCOM 2002. Twenty-First annual joint conference of the
  IEEE computer and communications societies. Proceedings. IEEE}, volume~3,
  pages 1597--1604. IEEE, 2002.

\bibitem{williamson2011design}
David~P Williamson and David~B Shmoys.
\newblock {\em The design of approximation algorithms}.
\newblock Cambridge University Press, 2011.

\bibitem{willsonbetter}
JK~Willson, L~Ding, W~Wu, L~Wu, Z~Lu, and W~Lee.
\newblock A better constant-approximation for coverage problem in wireless
  sensor networks.
\newblock {\em preprint}.

\bibitem{wu2013approximations}
Lidong Wu, Hongwei Du, Weili Wu, Deying Li, Jing Lv, and Wonjun Lee.
\newblock Approximations for minimum connected sensor cover.
\newblock In {\em INFOCOM, 2013 Proceedings IEEE}, pages 1187--1194. IEEE,
  2013.

\bibitem{wu2016connected}
Lidong Wu, Huijuan Wang, and Weili Wu.
\newblock Connected set-cover and group {S}teiner tree.
\newblock In {\em Encyclopedia of Algorithms}, pages 430--432. Springer, 2016.

\bibitem{wu2016minimum}
Lidong Wu and Weili Wu.
\newblock Minimum connected sensor cover.
\newblock In {\em Encyclopedia of Algorithms}, pages 1302--1304. Springer,
  2016.

\bibitem{zhang2012complexity}
Wei Zhang, Weili Wu, Wonjun Lee, and Ding-Zhu Du.
\newblock Complexity and approximation of the connected set-cover problem.
\newblock {\em Journal of Global Optimization}, 53(3):563--572, 2012.

\bibitem{zou2011new}
Feng Zou, Yuexuan Wang, Xiao-Hua Xu, Xianyue Li, Hongwei Du, Pengjun Wan, and
  Weili Wu.
\newblock New approximations for minimum-weighted dominating sets and
  minimum-weighted connected dominating sets on unit disk graphs.
\newblock {\em Theoretical Computer Science}, 412(3):198--208, 2011.

\end{thebibliography}

\end{document}